\documentclass{IEEEtran4PSCC}

\ifCLASSINFOpdf
   \usepackage[pdftex]{graphicx}
\else
   \usepackage[dvips]{graphicx}
\fi

\usepackage[cmex10]{amsmath}
\ifCLASSOPTIONcompsoc  \usepackage[caption=false,font=normalsize,labelfont=sf,textfont=sf]{subfig}
\else
  \usepackage[caption=false,font=footnotesize]{subfig}
\fi

\usepackage{cleveref}

\newcommand*{\cbra}[1]{\left\{ #1 \right\}}
\newcommand*{\rbra}[1]{\left( #1 \right)}
\newcommand*{\sbra}[1]{\left[ #1 \right]}

\newcommand*{\set}[2]{\left\{\, #1 \;\middle|\; #2 \,\right\}}
\newcommand*{\abs}[1]{\left\lvert #1 \right\rvert}

\DeclareMathOperator*{\argmin}{arg\,min}

\renewcommand{\epsilon}{\varepsilon}
\renewcommand{\theta}{\vartheta}
\renewcommand{\kappa}{\varkappa}
\renewcommand{\rho}{\varrho}
\renewcommand{\phi}{\varphi}

\newcommand{\prob}[1]{\mathbb{P}\left( #1 \right)}

\usepackage{url} 
\usepackage{eso-pic}
\usepackage{multirow} 
\usepackage{amssymb}
\usepackage{siunitx}
\usepackage{cite} 
\usepackage[english]{babel}

\usepackage{amsthm}
\theoremstyle{plain}
\newtheorem{theo}{Theorem}

\newtheorem*{defi}{Definition}
\crefname{theo}{theorem}{theorems}

\makeatletter
\let\old@ps@headings\ps@headings
\let\old@ps@IEEEtitlepagestyle\ps@IEEEtitlepagestyle
\def\psccfooter#1{%
    \def\ps@headings{%
        \old@ps@headings%
        \def\@oddfoot{\strut\hfill#1\hfill\strut}%
        \def\@evenfoot{\strut\hfill#1\hfill\strut}%
    }%
    \def\ps@IEEEtitlepagestyle{%
        \old@ps@IEEEtitlepagestyle%
        \def\@oddfoot{\strut\hfill#1\hfill\strut}%
        \def\@evenfoot{\strut\hfill#1\hfill\strut}%
    }%
    \ps@headings%
}
\makeatother

\psccfooter{%
        \parbox{\textwidth}{\hrulefill \\ \small{24th Power Systems Computation Conference} \hfill \begin{minipage}{0.2\textwidth}\centering \vspace*{4pt} \includegraphics[scale=0.06]{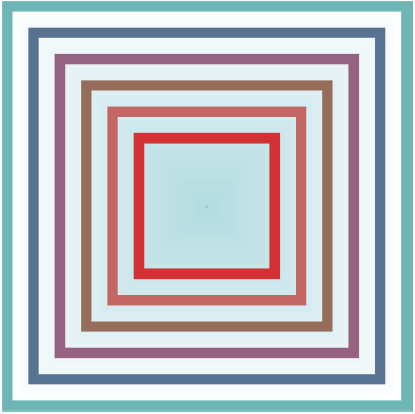}\\\small{PSCC 2026} \end{minipage} \hfill \small{Limassol, Cyprus --- June 8--12, 2026}}
}

\begin{document}

\title{Extreme Value Distributions of Peak Loads for Non-residential Customer Segments}

\author{\IEEEauthorblockN{Shaohong Shi\IEEEauthorrefmark{1},
 Eric A.~Cator\IEEEauthorrefmark{2}, 
 Jacco Heres\IEEEauthorrefmark{3} and 
 Simon H.~Tindemans\IEEEauthorrefmark{1}
 }
 \IEEEauthorblockA{\IEEEauthorrefmark{1} Dept.~of Electrical Sustainable Energy,
 Delft University of Technology,
 Delft, the Netherlands \\
 }
\IEEEauthorblockA{\IEEEauthorrefmark{2} Dept.~of Mathematics,
Radboud University, Nijmegen, the Netherlands \\
 }
 \IEEEauthorblockA{\IEEEauthorrefmark{3} Research Centre for Digital Technologies, 
 Alliander N.V.,
 Arnhem, the Netherlands \\ 
}
}

\maketitle

\begin{abstract}
Electrical grid congestion is a growing challenge in Europe, driving the need for accurate prediction of load, particularly of peak load. Non-time-resolved models of peak load offer the advantages of simplicity and compactness, and among them, Velander's formula (VF) is a traditional method that has been used for decades. Moreover, VF can be adapted into a quantile VF, which learns a truncated cumulative distribution function of peak load based on electricity consumption. This paper proposes a mathematical model based on extreme value theory to characterize the probability distribution of peak load for large non-residential customers. The model underpins the quantile VF as demonstrated through multiple quantile regression and reduces its representation to just four parameters without sacrificing predictive performance. Moreover, using maximum likelihood estimation and the likelihood ratio test, we validate that the probability distribution of peak load of analysed groups belongs to the heavy-tailed Fréchet class.
\end{abstract}


\thanksto{\noindent Submitted to the 24th Power Systems Computation Conference (PSCC 2026). This publication is part of the project ROBUST: Trustworthy AI-based Systems for Sustainable Growth with project number KICH3.LTP.20.006, which is partly financed by the Dutch Research Council (NWO).}

\section{Introduction} \label{sec-introduction}

Grid congestion in Europe is escalating. The cost incurred by EU transmission system operators for remedial actions to alleviate physical grid congestion is projected to rise to EUR 34\textendash  103 billion in 2040 \cite{Thomassen2024}. To manage congestion and reduce costs, it is essential to predict electrical loads and their associated metrics.

While time-resolved load models are all-encompassing by capturing full load behavior \cite{Khuntia2016}, non-time-resolved models that focus directly on peak load, a metric whose accuracy improvement brings the greatest economic benefit \cite{Ranaweera1997}, are typically simpler and more compact, making them more practical and easier for DSOs to deploy. Furthermore, they are indispensable in modeling peak loads for customers without smart meter data and for new customers. Examples include Velander's formula (VF) \cite{Velander1935}, Rusck's diversity factor (RDF) \cite{Rusck1956} and the simple form load model \cite{Seppaelae1996}.

Among them, VF is a traditional and well-established method to estimate peak load and aggregated peak load based on a customer's electricity consumption (EC), which is known or well-estimated for billing purposes. It was first used by Scandinavian distribution system operators (DSOs), and was later adopted by DSOs in the Netherlands and other European countries. Owing to its further development and practical application by Axelsson and Strand \cite{Axelsson1975}, the method is also referred to as the Strand-Axelsson formula. Case studies have shown its efficacy in providing reliable estimates of peak loads for individual customers \cite{Persson2018} and aggregations of customers \cite{Velander1935,Velander1952}, which makes it particularly useful in scenarios such as connecting new customers and reconfiguring substations.

In recent decades, it has become increasingly important for DSOs to model the uncertainty in loads \cite{Heres2017}, and there has been a trend towards probabilistic models. Besides probabilistic non-time-resolved models of peak load in \cite{Adams1991,Gibbons2014,Lee2022}, research has also explored the probability distribution of RDF, which also deals with peak load behavior \cite{Nazarko1998,Chatlani2007}.

To combine the advantages of VF and probabilistic modelling, authors of the present paper have proposed a quantile Velander formula (qVF) \cite{Shi2025}. The qVF is capable of learning truncated cumulative distribution functions (CDFs) of peak loads with multiple quantile regression (MQR). The qVF exhibited high efficacy in year-ahead prediction and reasonable efficacy in fitting customers with a wide range of ECs in the evaluation on three segments of large non-residential customers. Nevertheless, three limitations of the qVF remain to be addressed. First, with one parameter per quantile level, the qVF sacrifices the compactness of the original VF. Second, it fails to capture the tails of probability distributions of peak loads due to the inherent limitation in MQR. Finally, it does not explain the observed rightward shift of the CDFs of aggregated peak loads (corresponding to the same aggregated EC) as the aggregation level increases. 

As the qVF is an empirical formula, it is challenging to resolve its limitations without a supporting mathematical model. Accordingly, extreme value theory (EVT) emerges as a natural framework to analyze peak load \cite{Belzer1993,Li2020a,Lee2022,Jack2025}. In \cite{Lee2022}, the peak load is modeled to follow an extreme value distribution (EVD), whose parameters are related to the mean load by assuming the mean peak load to be a linear combination of the mean load and its square root. The assumption appears to be loosely inspired by the VF due to their structural similarity, which suggests a hidden connection between EVD and the VF.

\textbf{Contribution.} This paper proposes a mathematical model based on extreme value theory to characterize the probability distribution of peak load for large non-residential customers. The model provides a theoretical foundation for the quantile Velander formula and significantly reduces its parameter count from 82 to just four, which is demonstrated through multiple quantile regression. Furthermore, we apply maximum likelihood estimation and likelihood ratio tests to analyze the tail behavior of peak load distributions, identifying the Fréchet class as the most appropriate fit.

\section{Extreme value distribution model of peak load}

\textbf{Notation}. We use \(=\) to denote equality and \(:=\) to denote definition. We denote the set \(\cbra{1, 2, \ldots, n}\) by \([n]\) for \(n \in \mathbb{N}_{> 0}\).

\subsection{Problem statement}

Let \(T \in \mathbb{N}_{>0}\) and let \(\set{(P_i(t))_{t \in [T]}}{i \in C}\) be a set of load profiles for customers in class $\mathcal{C}$, where \([T]\) is interpreted as a discrete time range whose intervals last \(\Delta\). Define the peak load \(P^i_\mathrm{max}\) and the EC \(E_i\) of \(i \in C\) by
\begin{equation} \label{eq:PEdefinition}
P^i_\mathrm{max} := \max_{t \in [T]} P_i(t), \quad E_i := \sum_{t \in [T]} P_i(t) \Delta.
\end{equation}

Our objective is to characterize the probability distribution of the peak load \(P^i_\mathrm{max}\) for a customer \(i \in C\), as a function of the customer's energy consumption \(E_i\) (as in the VF and qVF). To this end, we introduce the following three assumptions.

\par{\textbf{A1}. There exists a positive integer \(K \le T\) and integers \(t_1, t_2, \ldots, t_K \in [T]\) such that for \(i \in C\), 
\begin{equation} P^i_\mathrm{max} = \max_{k \in [K]} P_i(t_k).
\end{equation}
}

\par{\textbf{A2}. For \(i \in C\), \(\set{P_i(t_k)}{k \in [K]}\) is a set of independent random variables.
}

\par{\textbf{A3}. There exist positive constants \(\theta_0, \theta_1\) and a probability distribution \(F\) such that for \(i \in C\) and \(k \in [K]\), 
\begin{equation} \label{eq-tilde-load} (P_i(t_k)-\theta_0 E_i)/\theta_1\sqrt{E_i} \sim F.
\end{equation}
}

Theoretically, we can set \(K = T\) in A1 to treat loads at all time points as i.i.d.\ random variables. However, this typically fails for realistic load profiles, as customers tend to have higher loads in peak time and lower loads in off-peak time. Instead, in A1, we assume that the peak load is the maximum of loads at \(K\) fixed time points, e.g.\ peak time during peak seasons, that these loads are independent in A2, and that these loads scale with their corresponding ECs in a manner inspired by the VF in A3.

We point out that none of these assumptions is strictly correct in the real world, but they aid in formulating a parametric model that will be shown to offer very good real-world performance.

\subsection{Model formulation}

We first briefly introduce relevant aspects of EVT, according to \cite{Haan2006}. The EVD \(G_\gamma\) with the extreme value index \(\gamma \in \mathbb{R}\) is defined by the CDF
\begin{equation} \label{eq-definition-EVD} G_\gamma(x) := \exp\rbra{-(1+\gamma x)^{-1/\gamma}}, \quad 1+\gamma x > 0.
\end{equation}
When \(\gamma = 0\), we interpret \((1+\gamma x)^{-1/\gamma}\) as \(\exp(-x)\) and \(G_\gamma\) is the CDF of the Gumbel distribution. When \(\gamma > 0\), \(G_\gamma\) belongs to the Fréchet class of distributions. When \(\gamma < 0\), \(G_\gamma\) belongs to the reverse-Weibull (r-Weibull) class of distributions.

We say that a distribution \(F\) is in the domain of attraction \(\mathcal{D}(G_\gamma)\), if there exists a positive function \(f\) and a real number \(\gamma\) such that for \(x\) with \(1+\gamma x>0\),
\begin{equation} \label{eq-assumption-EVT} \lim_{t \uparrow x^*} \frac{1-F(t+xf(t))}{1-F(t)} = (1+\gamma x)^{-1/\gamma},
\end{equation}
where \(F(x)\) is the CDF of \(F\) and \(x^* := \sup\set{x}{F(x)<1}\). Note that the Cauchy distribution is in \(\mathcal{D}(G_1)\), the normal, exponential and any gamma distribution are in \(\mathcal{D}(G_0)\), and the beta \((\mu,\nu)\) distribution is in \(\mathcal{D}(G_{-\mu^{-1}})\) \cite{Haan2006}. 

We now introduce the core theorem that will enable application of EVT to the peak load prediction problem.

\begin{theo} \label{theorem-EVT-individual} Let \(F\) be a distribution in \(\mathcal{D}(G_\gamma)\) for a \(\gamma \in \mathbb{R}\), and let \(\theta_0\) and  \(\theta_1\) be two positive constants. Let \(\set{E_i}{i \in C}\) be a set of non-negative numbers indexed by a set \(C\). For \(i \in C\), let \((X_i(t))_{t \in \mathbb{N}_{>0}}\) be a sequence of i.i.d.\ random variables with \(\tilde{X}_i(1) \sim F\), where
\begin{equation} \label{eq-theorem-EVT-i-tilde-X} \tilde{X}_i(t) := (X_i(t)-\theta_0 E_i)/\theta_1\sqrt{E_i}, \quad t \in \mathbb{N}_{>0}.
\end{equation}
For \(i \in C\) and \(n \in \mathbb{N}_{>0}\), let \(G^i_n(y)\) be the CDF of \(X^i_n\), where 
\begin{equation} \label{eq-theorem-EVT-i-Xin} X^i_n := \max_{t \in [n]} X_i(t).
\end{equation}
Then there exist a series \((a_n)_{n \in \mathbb{N}_{>0}}\) of positive real constants and a series \((b_n)_{n \in \mathbb{N}_{>0}}\) of real constants such that, for \(i \in C\),
\begin{equation} \label{eq-theorem-EVT-i-limit} \lim_{n \to \infty} \rbra{G^i_n(y) - G^{i,n}_\gamma(y)} = 0, \quad y \in S^{i,n}_\gamma,
\end{equation}
where for \(i \in C\) and \(n \in \mathbb{N}_{>0}\), 
\begin{equation} \label{eq-theorem-EVT-i-Gingamma} G_\gamma^{i,n}(y) := G_\gamma(z^i_n(y)), \quad y \in S^{i,n}_\gamma,
\end{equation}
\begin{equation} \label{eq-theorem-EVT-i-zin} z^i_n(y) :=  \frac{y}{\theta_1 a_n \sqrt{E_i}} - \frac{\theta_0 \sqrt{E_i}}{\theta_1 a_n} - \frac{b_n}{a_n}, \quad y \in S^{i,n}_\gamma,
\end{equation}
\begin{equation} \label{eq-theorem-EVT-i-Singamma} S^{i,n}_\gamma :=  \set{y \in \mathbb{R}}{1 +\gamma z^i_n(y) > 0}.
\end{equation}
Moreover, for \(i \in C\) and \(n \in \mathbb{N}_{>0}\), there is a probability distribution \(G_\gamma^{i,n}\) with CDF \(G_\gamma^{i,n}(y)\).

\end{theo}
\begin{proof}[Proof sketch] The theorem is proved by verifying the conditions of the extreme value theorem in Appendix. 
\end{proof}

To practically fit the EVD \(G_\gamma^{i,K}(y)\) to data, we require its quantile function (QF) and probability density function (PDF) in addition to its CDF.\footnote{Although \Cref{theorem-EVT-individual} does not formally guarantee the same convergence for the PDF and QF as for CDF, they are interchangeable for practical engineering purposes.} These are given in \Cref{theorem-EVT-CDF-PDF-QF}.

\begin{theo} \label{theorem-EVT-CDF-PDF-QF} Let \(i \in C\) and \(n \in \mathbb{N}_{>0}\). Denote by \(g^{i,n}_\gamma(y)\) and \(Q^{i,n}_\gamma(\tau)\) the PDF and QF of \(G_\gamma^{i,n}\). If \(\gamma = 0\), then for \(y \in \mathbb{R}\) and \(\tau \in (0,1)\),
\begin{equation} \label{eq-EVD-QF-G} Q_\gamma^{i,n}(\tau) = \theta_0 E_i + \theta_1(-a_K \ln(-\ln \tau)+b_K) \sqrt{E_i},
\end{equation}
\begin{equation} \label{eq-EVT-PDF-G} g_\gamma^{i,n}(y) = \frac{1}{\theta_1 a_K \sqrt{E_i}} \exp\rbra{-z_i(y)-\exp(-z_i(y))}.
\end{equation}
If \(\gamma \neq 0\), then for \(y \in S^{i,n}_\gamma\) and \(\tau \in (0,1)\),
\begin{equation} \label{eq-EVD-QF-FW} Q_\gamma^{i,n}(\tau) = \theta_0 E_i + \theta_1(a_K \gamma^{-1} \rbra{(-\ln \tau)^{-\gamma} - 1} + b_K) \sqrt{E_i},
\end{equation}
\begin{equation} \label{eq-EVD-PDF-FW} g_\gamma^{i,n}(y) = \frac{(1+\gamma z_i(y))^{-1-1/\gamma}}{\theta_1 a_K \sqrt{E_i}} \exp\rbra{-(1+\gamma z_i(y))^{-1/\gamma}}.
\end{equation}
\end{theo}
\begin{proof}[Proof sketch] The QF can be obtained by taking the inverse function of \(G^{i,n}_\gamma(y)\) and the PDF can be obtained by differentiating \(G^{i,n}_\gamma(y)\) with respect to \(y\).
\end{proof}

To apply \Cref{theorem-EVT-individual,theorem-EVT-CDF-PDF-QF} to the peak load \(P^i_\mathrm{max}\), we introduce the following two additional assumptions.

\par{\textbf{A4}. There is a \(\gamma \in \mathbb{R}\) such that \(F \in \mathcal{D}(G_\gamma)\).
}

\par{\textbf{A5}. The number \(K\) is sufficiently large such that the limit in \eqref{eq-theorem-EVT-i-limit} is effectively attained for accurate peak load estimation.
}

With Assumptions A1\textendash A4, we regard \(P_i(t_1)\), \(P_i(t_2)\), \(\ldots\), \(P_i(t_K)\) as the \(X_i(1)\), \(X_i(2)\), \(\ldots\), \(X_i(K)\) in \Cref{theorem-EVT-individual}. Together with Assumption A5, we approximate the CDF of \(P^i_\mathrm{max}\) for \(i \in C\) with \(G_\gamma^{i,K}(y)\). 

\begin{defi}[EVD peak load model] 
For a customer $i$ in class $\mathcal{C}$, with energy consumption $E_i$, we model the peak load as
\begin{equation}
    \hat{P}^i_\mathrm{max} \sim G_\gamma^{i,K}.
\end{equation}
The CDF, QF and PDF are given by \(G_\gamma^{i,K}(y)\), defined by \eqref{eq-theorem-EVT-i-Gingamma}, \(Q_\gamma^{i,K}(\tau)\), defined by \eqref{eq-EVD-QF-G} and \eqref{eq-EVD-QF-FW}, and \(g_\gamma^{i,K}(y)\), defined by \eqref{eq-EVT-PDF-G} and \eqref{eq-EVD-PDF-FW}, respectively.
\end{defi}

\subsection{Relation to the quantile Velander formula} \label{subs-relation-to-qVF}

Velander stated in \cite{Velander1935} that, if the load profiles are of the same type, there exist parameters \(\alpha\) and \(\beta\) corresponding to that type, such that the predicted peak load \(\hat{P}^i_\mathrm{max}\) is given by \begin{equation} \label{eq-Velander} \hat{P}^i_\mathrm{max} := \alpha E_i + \beta \sqrt{E_i}, \quad i \in C. \end{equation}
In \cite{Shi2025} the VF was extended to the qVF, which stated that, if the load profiles belong to the same class, there exist functions \(\alpha_\tau\) and \(\beta_\tau\) of \(\tau\) for that class, such that the predicted \(\tau\)-quantile \(\hat{P}^{i, \tau}_\mathrm{max}\) of \(P^i_\mathrm{max}\) is given by
\begin{equation} \label{eq-Velander-quantile} \hat{P}^{i, \tau}_\mathrm{max} := \alpha_\tau E_i + \beta_\tau \sqrt{E_i}, \quad i \in C, \tau \in (0,1).
\end{equation}
In \cite{Shi2025}, \(\tau\) was limited to a set \(A\) of discrete quantile levels to enable MQR to learn the parameters, with the following constraint (labeled `C4') for \(\alpha_\tau\) and \(\beta_\tau\): 
\begin{equation} \label{eq-C4} \alpha_{\tau_1} = \alpha_{\tau_2}\equiv \alpha, \quad \beta_{\tau_1} \le \beta_{\tau_2}, \quad \tau_1 \le \tau_2, \;\; \tau_1, \tau_2 \in A.\end{equation}

Both the VF and the qVF are empirical formulas. However, if we equate the \(\hat{P}^{i, \tau}_\mathrm{max}\) in \eqref{eq-Velander-quantile} with the quantile \(Q^{i,K}_\gamma(\tau)\) of \(P^i_\mathrm{max}\) given in \eqref{eq-EVD-QF-G} and \eqref{eq-EVD-QF-FW}, then we obtain the following parametric representation of \(\alpha_\tau\) and \(\beta_\tau\) for \(\tau \in (0, 1)\):
\begin{align} \label{eq-parametric-representation-alpha} \alpha_\tau = & \; \theta_0, \\
\label{eq-parametric-representation-beta} \beta_\tau = & \begin{cases}
    \theta_1(-a_K \ln(-\ln \tau)+b_K), & \text{ if } \gamma = 0, \\
    \theta_1(a_K \gamma^{-1} \rbra{(-\ln \tau)^{-\gamma} - 1} + b_K), & \text{ if } \gamma \neq 0. 
\end{cases}
\end{align}
This shows the potential of the EVD model to support the empirical qVF, reducing the number of parameters in the representation of \(\alpha_\tau\) and \(\beta_\tau\) from \(1+\abs{A}\) in the qVF to only four: \(\theta_0\), \(\theta_1 a_K\), \(\theta_1 b_K\) and \(\gamma\).

\section{Parameter estimation methodology}

We apply both MQR and MLE to evaluate the EVD model. MQR is used under a similar setting as in \cite{Shi2025} to assess whether the EVD model underpins the qVF by effectively representing the parameters of the qVF with \eqref{eq-parametric-representation-alpha} and \eqref{eq-parametric-representation-beta}. In parallel, MLE is applied together with the likelihood ratio test (LRT) to identify the tail behavior of the probability distributions of peak loads.

When performing MQR and MLE, \(C_\mathrm{tr}\) is the subset of indices of training data, and \(w = (w_0, w_1, \ldots, w_{d-1})\) and \(W\) refers to the parameter vector and the parameter space, respectively, where \(d\) is the dimension of \(w\). Note that we close the range of parameters at their finite end(s) when performing MQR and MLE to avoid excluding potential minimizers located at the boundary of the parameter space. For example, we replace \(w_0 > 0\) with \(w_0 \ge 0\).

\subsection{Multiple quantile regression} \label{subs-analyse-MQR}

The loss function in MQR is the summation of pinball losses at different quantile levels plus an optional regularizer \cite{Takeuchi2006}. The loss function generalizes the pinball loss, and was proven to be a proper scoring rule in \cite{Cervera1996}.

Let \(A\) be a finite set of distinct quantile levels, and let \(Q^{i}(\tau; w)\) denote the predicted value of \(P^{i,\tau}_\mathrm{max}\) given parameter vector \(w\), for \(i \in C\) and \(\tau \in A\). The MQR for obtaining the optimal parameter vector \(\hat{w}\) is given by 
\begin{equation} \label{eq-MQR} \hat{w} := \argmin_{w \in W} \frac{1}{|C_\mathrm{tr}|} \frac{1}{|A|} \sum_{i \in C_\mathrm{tr}} \sum_{\tau \in A} \mathrm{PL}^{i}(\tau; w),\end{equation}
where the objective is the \emph{average pinball loss} (APL) and
\begin{equation} \mathrm{PL}^{i}(\tau; w) := \begin{cases} (\tau-1) \delta^i(\tau; w) & \text{if } \delta^i(\tau; w) < 0, \\
\tau \delta^i(\tau; w) & \text{if } \delta^i(\tau; w) \ge 0, \end{cases}\end{equation}
\begin{equation} \delta^i(\tau; w) := P^i_\mathrm{max} - Q^{i}(\tau; w).
\end{equation}
We omit regularization in \eqref{eq-MQR} since the functions \(Q^{i}(\tau; w)\) that we will use are structurally well-defined and exhibit monotonicity by design.

We perform the MQR \eqref{eq-MQR} in the following five formulations separately: C4 (the qVF under the constraint \eqref{eq-C4}), Gumbel, f-Gumbel (fuzzy-Gumbel), Fréchet and r-Weibull. The \(w\) and \(W\) of each formulation are defined in the second and the third columns of \Cref{table-mqr-mle-formulations}, respectively, where \(\gamma_\mathrm{th}\) is a small postive constant. Note that the \(Q^{i}(\tau; w)\) of the C4 is reparameterized from \eqref{eq-Velander-quantile}, and the \(Q^{i}(\tau; w)\) of Gumbel, Fréchet and r-Weibull formulations are reparameterized from the formulas for \(Q^{i,K}_\gamma(\tau)\) given in \Cref{theorem-EVT-CDF-PDF-QF}.

We justify the f-Gumbel formulation. As \(\gamma\) approaches zero, the expression \(\gamma^{-1} \rbra{(-\ln \tau)^{-\gamma} - 1}\) in \eqref{eq-EVD-QF-FW} becomes ill-conditioned due to cancellation effects and division by small values. To mitigate this, we employ the Taylor polynomial of degree \(3\) at \(\gamma = 0\) of the right-hand side (RHS) of \eqref{eq-EVD-QF-FW} for a numerically stable approximation for \(\gamma \in [-\gamma_\mathrm{th}, \gamma_\mathrm{th}]\), i.e.\
\begin{align} \label{eq-MQR-f-Gumbel-h}
Q^{i}(\tau; w) := & \; w_0 E_i + (w_1 (-\ln(-\ln \tau)  + w_3 (\ln(-\ln \tau))^2/2 \notag \\
& - w_3^2 (\ln(-\ln \tau))^3/6 + w_3^3 (\ln(-\ln \tau))^4/24) \notag \\
& + w_2) \sqrt{E_i}.
\end{align}

\begin{table*}[!htbp]
\renewcommand{\arraystretch}{1.3}
\caption{Parameter vector \(w\) and parameter space \(W\) in different formulations of MQR and MLE}
\begin{center}
\begin{tabular}{|c||c|c||c|c|} \hline
Form. & \(w\) in MQR & \(W\) in MQR & \(w\) in MLE & \(W\) in MLE \\ \hline
C4 & \(((\alpha_\tau, \beta_\tau))_{\tau \in A}\) & \(\set{w}{w \text{ satisfies } \eqref{eq-C4}}\) & \multicolumn{2}{c|}{not applicable}  \\ \hline
Gumbel & \((\theta_0, \theta_1 a_K, \theta_1 b_K)\) & \(\mathbb{R}_{\ge 0} \times \mathbb{R}_{\ge 0} \times \mathbb{R}\) & \(\left(\frac{1}{\theta_1 a_K}, \frac{\theta_0}{\theta_1 a_K}, \frac{b_K}{a_K}\right)\) & \(W_0 := \mathbb{R}_{\ge 0} \times \mathbb{R}_{\ge 0} \times \mathbb{R}\)  \\ \hline
f-Gumbel & \((\theta_0, \theta_1 a_K, \theta_1 b_K, \gamma)\) & \(\mathbb{R}_{\ge 0} \times \mathbb{R}_{\ge 0} \times \mathbb{R} \times [-\gamma_\mathrm{th}, \gamma_\mathrm{th}]\) & \multirow{3}{*}{\(\left(\frac{1}{\theta_1 a_K}, \frac{\theta_0}{\theta_1 a_K}, \frac{b_K}{a_K}, \gamma\right)\)} & \(W_0 \times [-\gamma_\mathrm{th}, \gamma_\mathrm{th}]\) \\ \cline{1-3} \cline{5-5}
Fréchet & \multirow{2}{*}{\((\theta_0, \theta_1 a_K \gamma^{-1}, \theta_1 (b_K - a_K \gamma^{-1}), \gamma)\)} & \(\mathbb{R}_{\ge 0} \times \mathbb{R}_{\ge 0} \times \mathbb{R} \times [\gamma_\mathrm{th}, \infty)\) &  & \(W_0 \times [\gamma_\mathrm{th}, \infty) \cap W^*\) \\ \cline{1-1} \cline{3-3} \cline{5-5}
r-Weibull & & \(\mathbb{R}_{\ge 0} \times \mathbb{R}_{\le 0} \times \mathbb{R} \times (-\infty, -\gamma_\mathrm{th}]\) & & \(W_0 \times (-\infty, -\gamma_\mathrm{th}] \cap W^*\) \\ \hline
\end{tabular}
\label{table-mqr-mle-formulations}
\end{center}
\end{table*}

\subsection{Maximum likelihood estimation} \label{subs-analyse-MLE}

With the PDF of peak loads given in \Cref{theorem-EVT-CDF-PDF-QF}, the EVD model can also be fitted with MLE in addition to MQR. As MQR is limited to a finite set of discrete quantile levels, its ability to capture tail behavior is inherently constrained. In contrast, MLE leverages the full probability distribution, enabling more effective modeling of the tails. In this section, we formulate MLE to evaluate the EVD model and the likelihood ratio test (LRT) to identify the tail behavior of the probability distribution of peak load. Additionally, we calculate the approximate standard deviation of the estimated \(\gamma\) via MLE in the Fréchet formulation, which provides insight into the result from the LRT.

The MLE for obtaining the optimal parameter vector \(\hat{w}\) is given by 
\begin{equation} \label{eq-MLE} \hat{w} := \argmin_{w \in W} \frac{1}{|C_\mathrm{tr}|} \sum_{i \in C_\mathrm{tr}} -\ell^i(P^i_\mathrm{max}; w),\end{equation}
where \(\ell^i(P^i_\mathrm{max}; w)\) is the log-likelihood of observation $P^i_\mathrm{max}$ of customer \(i\) for parameter vector \(w\). The objective is the average negative log-likelihood (ANLL) over the training set. 

We perform the MLE \eqref{eq-MLE} in the following four formulations separately: Gumbel, f-Gumbel (fuzzy-Gumbel), Fréchet and r-Weibull. Note that the qVF cannot be fitted with MLE due to its lack of a complete underlying distribution. The \(w\) and \(W\) of each formulation are defined in the fourth and the fifth columns of \Cref{table-mqr-mle-formulations}, respectively, where \(W^*\) originates from \eqref{eq-theorem-EVT-i-Singamma} and is defined by
\begin{equation} \label{eq-def-W*} W^* := \set{w \in \mathbb{R}^4}{1+w_3 z_i > 0 \text{ for } i \in C},
\end{equation}
\begin{equation} \label{eq-EVT-MLE-zi} z_i := w_0 P^i_\mathrm{max}/\sqrt{E_i} - w_1 \sqrt{E_i} - w_2, \quad i \in C.
\end{equation}

We introduce \(\ell^i(P^i_\mathrm{max}; w)\) in each formulation, which is reparameterized from \(\ln g_\gamma^{i, K}\rbra{P^i_\mathrm{max}}\) in general terms. Namely, in the Gumbel formulation, 
\begin{equation} \ell^i(P^i_\mathrm{max}; w) := \ln w_0 - \frac{1}{2} \ln E_i - z_i - \exp(-z_i),
\end{equation}
and in the Fréchet and r-Weibull formulations, 
\begin{align} \label{eq-MLE-NLL-FW-original} \ell^i(P^i_\mathrm{max}; w) := & \ln w_0 - \frac{1}{2} \ln E_i - (1+w_3 z_i)^{-1/w_3}  \notag \\
& - (1+1/w_3)\ln(1+w_3 z_i).
\end{align}

Two issues arise during the MLE. First, when \(\gamma\) is close to \(0\), the last two terms in the RHS of \eqref{eq-MLE-NLL-FW-original} become ill-conditioned. To mitigate this, the f-Gumbel formulation is proposed for \(\gamma \in [-\gamma_\mathrm{th}, \gamma_\mathrm{th}]\), where \(\gamma_\mathrm{th}\) is a small positive constant and the Taylor polynomial of degree \(2\) at \(\gamma = 0\) of the RHS of \eqref{eq-MLE-NLL-FW-original} is used as a numerically stable approximation, i.e.\ 
\begin{align} \label{eq-MLE-NLL-FW-Taylor} \ln w_0 & - \frac{1}{2} \ln E_i - z_i - \exp(-z_i) \notag \\
& - w_3 \left( -\frac{z_i^2}{2} + \frac{z_i^2 \exp(-z_i)}{2} + z_i \right) \notag \\
& - w_3^2 \left( \frac{z_i^4 \exp(-z_i)}{8} + \frac{z_i^3}{3} - \frac{z_i^3 \exp(-z_i)}{3} - \frac{z_i^2}{2} \right).
\end{align}

The other issue is that during objective function evaluations in the optimization, \(1+w_3 z_i\) in \eqref{eq-MLE-NLL-FW-original} can be a very small positive number or even negative. To ensure numerical stability of the evaluations, we replace the logarithm and power operations with their numerically safe counterparts in the optimization. Specifically, we replace the RHS of \eqref{eq-MLE-NLL-FW-original} with
\begin{align} \label{eq-MLE-NLL-FW} \ln w_0 & - \frac{1}{2} \ln E_i  - \max\rbra{1+w_3 z_i, \epsilon_\mathrm{th}}^{-1/w_3} \notag\\
& - (1+1/w_3)\ln\max\rbra{1+w_3 z_i, \epsilon_\mathrm{th}},
\end{align}
where \(\epsilon_\mathrm{th}\) is a small positive constant. Note that this replacement is used only in the optimization with the Fréchet and r-Weibull formulations, not in other scenarios such as calculating the training and testing ANLLs.

Next, we introduce the LRT to compare the Gumbel and the Fréchet formulations \cite{Coles2001}. From the testing ANLLs shown later in \Cref{table-anll-evt}, it is not obvious whether the simpler Gumbel formulation suffices, or the Fréchet formulation provides a significantly better fit with an additional parameter \(\gamma \neq 0\). To investigate this, we define the hypothesis test with
\begin{itemize}
  \item the null hypothesis \(H_0\): \(\gamma = 0\), and
  \item the alternative hypothesis \(H_1\): \(\gamma \ge \gamma_\mathrm{th}\).
\end{itemize}
Let \(-\ell_0\) and \(-\ell_1\) be the minimums of \(\sum_{i \in C} -\ell^i(P^i_\mathrm{max}; w)\) under \(H_0\) (the Gumbel formulation)  and \(H_1\) (the Fréchet formulation), respectively. Then the test statistic is \(\Lambda = -2(\ell_0 -\ell_1)\) and the \(p\)-value is computed as \(1- \chi^2_1(\Lambda)\), where \(\chi^2_1\) is the CDF of the \(\chi^2\)-distribution with degrees of freedom equal to \(1\). A small \(p\)-value indicates that the EVD model with \(\gamma > \gamma_\mathrm{th}\) provides a significantly better fit than that with \(\gamma = 0\), justifying the inclusion of the extreme value index \(\gamma\).

Finally, we approximate the standard deviation \(\mathrm{Std}\rbra{\hat{\gamma}}\) of the estimated \(\gamma\) in the Fréchet formulation, by taking the square root of the relevant entry of the inverse of the observed Fisher information at the estimated \(\gamma\) \cite{Casella2024}. The Fisher information is the Hessian matrix of \(-\ell_1(w) = \sum_{i \in C} -\ell^i(P^i_\mathrm{max}; w) \) evaluated at the estimate \(\hat{w} := \argmin_{w \in W} -\ell_1(w)\), so that $\mathrm{Std}\rbra{\hat{\gamma}} \approx \sqrt{\sbra{\rbra{H(-\ell_1(\hat{w}))}^{-1}}_{4,4}}$.


\section{Data, results and analysis} 
Smart meter data of large electricity customers (with a grid capacity between 60\SI{}{\kilo\watt} and 100\SI{}{\mega\watt}) for the years 2022, 2023 and 2024 were collected from Dutch DSO Liander with a 15-minute resolution. The units of loads and ECs are \SI{}{\kilo\watt} and kW\(\cdot\)15~minute, respectively, where \(\Delta\) is set to 15 minutes so that its value equals \(1\) in \eqref{eq:PEdefinition}. We performed analyses on three categories of large customers: SBI code 8411 (general government administration), SBI code 6420 (financial holding companies) and KvK code 004 (industry). \Cref{tabl-number-customers} shows the number of customers in the original data and the number of customers after removing customers whose load profiles are incomplete, have negative values, and are all zero at the first 672 time points (during the first week).

\begin{table}[!ht]
\caption{Number of customers (original \(\rightarrow\) processed)}
\begin{center}
\begin{tabular}{|c|c|c|c|}
\hline
\textbf{Segment\textbackslash\,Year} & 2022 & 2023 & 2024 \\ \hline
SBI code 8411 & 1058 \(\rightarrow\) 778 & 1197\(\rightarrow\) 794 & 1202 \(\rightarrow\) 873 \\ \hline
SBI code 6420 & 1065 \(\rightarrow\) 709 & 1143 \(\rightarrow\) 727 & 1144 \(\rightarrow\) 765 \\ \hline
KvK code 004 & 1261 \(\rightarrow\) 952 & 1324 \(\rightarrow\) 934 & 1325 \(\rightarrow\) 961 \\ \hline
\end{tabular}
\label{tabl-number-customers}
\end{center}
\end{table}

In model fitting, we set \(\gamma_\mathrm{th} = \SI{1e-2}{}\) and \(\epsilon_\mathrm{th} = \SI{1e-20}{}\) and adopt the same set \(A\) of quantile levels as in \cite{Shi2025}, i.e.\ \(\cbra{0.10,0.11, \ldots, 0.90}\). To derive a more accurate estimate of model performance, we conducted 5-fold cross-validation for all formulations, and report their average training and testing APLs from MQR in \Cref{table-apl-EVT} and ANLLs from MLE in \Cref{table-anll-evt}. Note that lower values of APL and ANLL indicate a better model fit. In the two tables we also report the average estimated values \(\hat{\gamma}\) of \(\gamma\) across all folds for the f-Gumbel, Fréchet and r-Weibull formulations. The last column of \Cref{table-anll-evt} presents the standard deviation of \(\hat{\gamma}\) in the Fréchet formulation. To intuitively show the model, we plot the quantile functions obtained via MLE in the Fréchet formulation for the customers with SBI code 8411 in 2022 in \Cref{fig-qr-curves-2022-SBI-8411}. Both peak load and EC are plotted on a logarithmic scale to improve visibility.

\begin{table*}[!htbp]
\caption{Multiple Quantile Regression: Average training (Tr) APLs (kW), testing (Te) APLs (kW) and \(\hat{\gamma}\)}
\begin{center}
\begin{tabular}{|c||c|c||c|c||c|c|c||c|c|c||c|c|c|} \hline
\multirow{2}{*}{Year} & \multicolumn{2}{c||}{C4} & \multicolumn{2}{c||}{Gumbel} & \multicolumn{3}{c||}{f-Gumbel} & \multicolumn{3}{c||}{Fréchet} & \multicolumn{3}{c|}{r-Weibull} \\ 
 & Tr & Te & Tr & Te & Tr & Te & \(\hat{\gamma}\) & Tr & Te & \(\hat{\gamma}\) & Tr & Te & \(\hat{\gamma}\) \\ \hline
\multicolumn{14}{|c|}{SBI code 8411} \\ \hline
2022 & 30.81 & 31.17 & 30.93 & 31.22 & 30.92 & 31.22 & 0.01000 & 30.81 & 31.16 & 0.45978 & 30.93 & 31.22 & -0.01044 \\ \hline
2023 & 28.83 & 30.13 & 28.92 & 30.12 & 28.91 & 30.11 & 0.01000 & 28.84 & 30.09 & 0.39111 & 28.92 & 30.12 & -0.01000 \\ \hline
2024 & 27.93 & 28.09 & 28.03 & 28.17 & 28.03 & 28.16 & 0.01000 & 27.94 & 28.08 & 0.44643 & 28.04 & 28.17 & -0.01000 \\ \hline
\multicolumn{14}{|c|}{SBI code 6420} \\ \hline
2022 & 24.30 & 29.22 & 24.33 & 29.25 & 24.33 & 29.25 & 0.01000 & 24.31 & 29.22 & 0.18692 & 24.33 & 29.26 & -0.01016 \\ \hline
2023 & 24.66 & 27.13 & 24.68 & 27.14 & 24.68 & 27.14 & 0.00930 & 24.67 & \textbf{27.16} & 0.11871 & 25.47 & 27.95 & -0.01528 \\ \hline
2024 & 29.27 & 30.82 & 29.29 & 30.88 & 29.29 & 30.88 & 0.00600 & 29.28 & 30.87 & 0.11820 & 29.29 & 30.88 & -0.01617 \\ \hline
\multicolumn{14}{|c|}{KvK code 004} \\ \hline
2022 & 42.93 & 44.25 & 42.95 & 44.19 & 42.95 & 44.19 & 0.01000 & 42.94 & \textbf{44.25} & 0.08645 & 42.95 & 44.20 & -0.01133 \\ \hline
2023 & 43.00 & 43.52 & 43.04 & 43.51 & 43.03 & 43.50 & 0.01000 & 43.01 & 43.49 & 0.14750 & 43.04 & 43.51 & -0.01001 \\ \hline
2024 & 41.33 & 41.57 & 41.36 & 41.54 & 41.35 & 41.54 & 0.01000 & 41.34 & 41.53 & 0.14799 & 41.36 & 41.54 & -0.01044 \\ \hline
\end{tabular}
\label{table-apl-EVT}
\end{center}
\end{table*}

\begin{table*}[!htbp]
\caption{Maximum Likelihood Estimation: Average training (Tr) ANLLs, testing (Te) ANLLs, \(\hat{\gamma}\) and \(\mathrm{Std}\rbra{\hat{\gamma}}\)}
\begin{center}
\begin{tabular}{|c||c|c||c|c|c||c|c|c|c||c|c|c||} \hline
\multirow{2}{*}{Year} & \multicolumn{2}{c||}{Gumbel} & \multicolumn{3}{c||}{f-Gumbel} & \multicolumn{4}{c||}{Fréchet} & \multicolumn{3}{c||}{r-Weibull} \\ 
  & Tr & Te & Tr & Te & \(\hat{\gamma}\) & Tr & Te & \(\hat{\gamma}\) & \(\mathrm{Std}\rbra{\hat{\gamma}}\) & Tr & Te & \(\hat{\gamma}\) \\ \hline
\multicolumn{13}{|c|}{SBI code 8411} \\ \hline
2022 & 5.4390 & 5.4435 & 5.4261 & 5.4304 & 0.01000 & 5.2654 & 5.2693 & 0.44910 & 0.03774 & 5.4536 & 5.4583 & -0.01000 \\ \hline
2023 & 5.3505 & 5.3672 & 5.3379 & 5.3539 & 0.01000 & 5.1863 & 5.1960 & 0.40683 & 0.03387 & 5.3645 & 5.3821 & -0.01000 \\ \hline
2024 & 5.2756 & 5.2861 & 5.2574 & 5.2660 & 0.01000 & 5.1151 & 5.1224 & 0.34578 & 0.03038 & 5.3008 & 5.3171 & -0.01000 \\ \hline
\multicolumn{13}{|c|}{SBI code 6420} \\ \hline
2022 & 5.3140 & 5.3752 & 5.3019 & 5.3682 & 0.01000 & 5.1815 & \textbf{5.3798} & 0.33371 & 0.03398 & 5.3278 & 5.3837 & -0.01000 \\ \hline
2023 & 5.3392 & 5.3551 & 5.3192 & 5.3335 & 0.01000 & 5.1716 & 5.1826 & 0.34270 & 0.03363 & 5.3687 & 5.3916 & -0.01000 \\ \hline
2024 & 5.2720 & 5.2819 & 5.2575 & 5.2660 & 0.01000 & 5.1351 & 5.1401 & 0.32205 & 0.03271 & 5.2896 & 5.3018 & -0.01000 \\ \hline
\multicolumn{13}{|c|}{KvK code 004} \\ \hline
2022 & 5.6688 & 5.6712 & 5.6634 & 5.6657 & 0.01000 & 5.6315 & 5.6352 & 0.17021 & 0.02478 & 5.6750 & 5.6775 & -0.01000 \\ \hline
2023 & 5.6463 & 5.6521 & 5.6422 & 5.6477 & 0.01000 & 5.6207 & 5.6307 & 0.14334 & 0.02459 & 5.6509 & 5.6572 & -0.01000 \\ \hline
2024 & 5.5770 & 5.5820 & 5.5729 & 5.5777 & 0.01000 & 5.5536 & 5.5589 & 0.12916 & 0.02289 & 5.5820 & 5.5870 & -0.01000 \\ \hline
\end{tabular}
\label{table-anll-evt}
\end{center}
\end{table*}

\begin{figure}[!ht]
\centering
\includegraphics[width=0.9\linewidth]{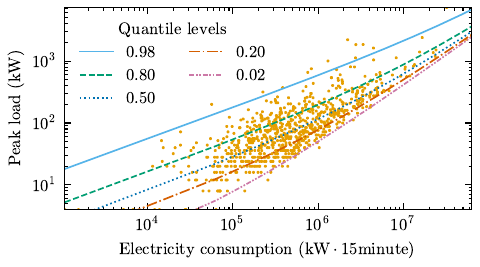}
\caption{SBI code 8411 in 2022: peak loads and electricity consumption of customers alongside curves of quantile functions at different quantile levels, for the MLE-fitted Fr\'echet model.}
\label{fig-qr-curves-2022-SBI-8411}
\end{figure}

We next verify the parametric representation in \eqref{eq-parametric-representation-beta} and compare the Gumbel and Fréchet formulations. The r-Weibull formulation is excluded from further analysis, as its estimated \(\gamma\) is either close to or equal to the upper bound \(\SI{-1e-2}{}\) of the range of \(\gamma\), and according to the APLs and ANLLs, it fails to provide a better fit with the additional parameter \(\gamma\) compared to the Gumbel formulation. 

\subsection{Parametric representation}

As shown in \Cref{table-apl-EVT} and \Cref{table-anll-evt}, the C4, Gumbel, and Fréchet formulations yield similar testing APLs and ANLLs. Note that their number of parameters are \(82\), \(3\) and \(4\), respectively. This strongly indicates that \eqref{eq-parametric-representation-alpha} and \eqref{eq-parametric-representation-beta} provide a valid parametric representation of the parameters in the qVF under the constraint \eqref{eq-C4}.

To illustrate this intuitively, we plot the set of discrete points \(\set{(\tau,\beta_\tau)}{\tau \in A}\) with parameters fitted from the C4 formulation, alongside continuous curves \(\set{(\tau,\beta_\tau)}{\tau \in (0.01, 0.99)}\), where \(\beta_\tau\) is calculated from \eqref{eq-parametric-representation-beta} with parameters fitted from the Gumbel and Fréchet formulation via either MQR or MQR, from the  formulation via MQR, from the Gumbel formulation via MLE, from the Fréchet formulation via MLE, respectively. Note that the parameters here are fitted to the entire set \(C\). In \Cref{fig-betas-comparison} we show the results for the customers with SBI code 8411 in 2022, with SBI code 6420 in 2023 and with KvK code 004 in 2024, respectively.

\begin{figure*}[!htbp]
\centering
\subfloat[SBI 8411 in 2022]{
    \includegraphics[width=0.31\linewidth]{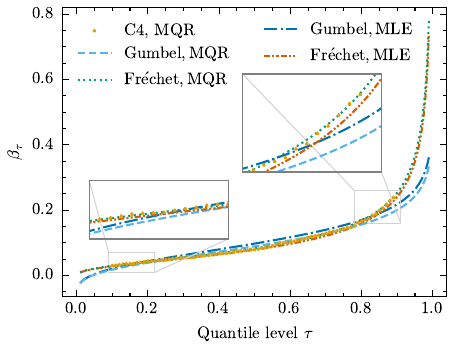}
    \label{fig-mqr-betas-SBI-8411-2022}
}
\hfill
\subfloat[SBI 6420 in 2023]{
    \includegraphics[width=0.31\linewidth]{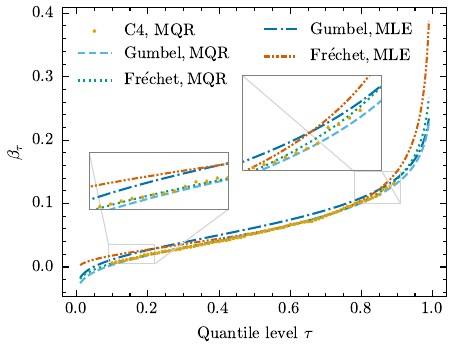}
    \label{fig-mqr-betas-SBI-6420-2023}
}
\hfill
\subfloat[KvK 004 in 2024]{
    \includegraphics[width=0.31\linewidth]{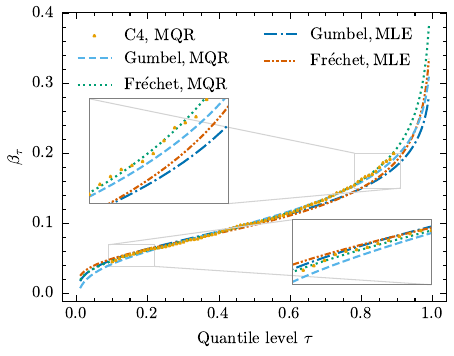}
    \label{fig-mqr-betas-KvK-004-2024}
}
\caption{Comparison of $\beta_\tau$ from the qVF under C4, from the Gumbel formulation and from the Fréchet formulation.}
\label{fig-betas-comparison}
\end{figure*}

As shown in the figures, the parametric representation \eqref{eq-parametric-representation-beta} that uses parameters estimated from the Fréchet formulation via MQR, closely reconstructs the \(\beta_\tau\) values in the qVF under the constraint C4 across the entire set \(A\) of \(\tau\). By comparison, the same representation with parameters from the Gumbel formulation via MQR exhibits deviations for \(\tau\) close to \(0\) or \(1\), which is more evident in \Cref{fig-mqr-betas-SBI-8411-2022} than in \Cref{fig-mqr-betas-SBI-6420-2023,fig-mqr-betas-KvK-004-2024}. Additionally, there are minor differences between the results via MQR and MLE, which is expected given that MLE incorporates the full distribution and is more sensitive to outliers.

\subsection{Gumbel or Fréchet}

As shown in \Cref{table-apl-EVT} and \Cref{table-anll-evt}, the Fréchet formulation generally yields slightly lower testing APLs and ANLLs than the Gumbel formulation, except for the bolded entries. However, this advantage is marginal, especially considering the additional parameter \(\gamma\) and its near-zero values in the Fréchet formulation.

In contrast, the computed \(p\)-values from the LRTs were less than \(\SI{1e-10}{}\) for all the three customer categories across the three years, and they indicate that the Fréchet formulation performs significantly better than the Gumbel formulation from a statistical perspective. This finding does not contradict the near-zero estimates of \(\gamma\), as the even smaller standard deviation \(\mathrm{Std}\rbra{\hat{\gamma}}\) suggests that \(\gamma\) is consistently estimated to be close to, but not exactly, zero. In addition, this statistical advantage is supported by the MQR results in \Cref{fig-mqr-betas-SBI-8411-2022,fig-mqr-betas-SBI-6420-2023,fig-mqr-betas-KvK-004-2024}, which suggest that the Fréchet formulation reconstructs the \(\beta_\tau\) values in the qVF slightly better than the Gumbel formulation. DSOs may additionally prefer the Fréchet formulation because its heavier tail better captures the possibility of outliers, thus resulting in more conservative operation. 

\section{Conclusion}

We proposed an EVT-based model of peak electricity load of customers. Although the model was derived on the basis of idealised assumptions, evaluation on real-world data of three categories of non-residential customers indicates that the model can effectively supplant previous formulations of the quantile Velander's formula, providing a compact and theoretically sound approach to peak load forecasting. Additionally, maximum likelihood estimation, in combination with the likelihood ratio test, identified the Fréchet class of heavy-tailed distributions as the best fit for the distribution of peak load behavior, which provides valuable insights for DSOs. 

Ongoing research focuses extending the model and fitting methods to capture correlation among customers' aggregated peak loads more effectively. This will ideally enable a single high-accuracy model to be used for single customers within a class as well as aggregations of customers. Such models have immediate applications in network reinforcement deferral and connection planning for new customers.

\section*{Acknowledgment}
 
The authors thank Han La Poutr\'e for helpful comments. Microsoft Copilot was used to assist with minor language adjustments under author supervision.


\section*{Appendix: Proof of \Cref{theorem-EVT-individual}} \label{sec-appendix}

Let \(\tilde{X}_n\) denote the maximum of \(n\) i.i.d.\ random variables with distribution \(F\) for \(n \in \mathbb{N}\). Together with \(F \in \mathcal{D}(G_\gamma)\), it follows from Theorem 1.1.6 of \cite{Haan2006} that, there exist a series \((a_n)_{n \in \mathbb{N}_{>0}}\) of positive real numbers and a series \((b_n)_{n \in \mathbb{N}_{>0}}\) of real numbers such that, for \(x\) with \(1+\gamma x > 0\),
\begin{equation} \label{eq-limit-EVT-star} \lim_{n \to \infty} \prob{\tilde{X}_n \le a_n x + b_n} = G_\gamma(x).
\end{equation}

Using \eqref{eq-theorem-EVT-i-Xin} and \eqref{eq-theorem-EVT-i-tilde-X}, we define, for \(i \in C\) and \(n \in \mathbb{N}_{>0}\), the scaled maximum \(\tilde{X}^i_n :=  (X^i_n-\theta_0 E_i)/\theta_1\sqrt{E_i}\), with \(\tilde{X}^i_n = \max_{t \in [n]} \tilde{X}_i(t)\). 
Then, by the conditions imposed on \(\tilde{X}_i(t)\), for \(i \in C\), \(\tilde{X}^i_n\) is identically distributed as \(\tilde{X}_n\). Thus, \eqref{eq-limit-EVT-star} holds for \(i \in C\). Namely, with the given series \((a_n)_{n \in \mathbb{N}_{>0}}\) and \((b_n)_{n \in \mathbb{N}_{>0}}\) of constants, for \(i \in C\) and \(x\) with \(1+\gamma x > 0\),
\begin{equation} \label{eq-limit-EVT} \lim_{n \to \infty} \rbra{\prob{\tilde{X}^i_n \le a_n x + b_n} - G_\gamma(x)} = 0.
\end{equation}
Then we obtain \eqref{eq-theorem-EVT-i-limit} by replacing \(x\) with \(z^i_n(y)\) and replacing \(\tilde{X}^i_n\) with \((X^i_n-\theta_0 E_i)/\theta_1\sqrt{E_i}\) in \eqref{eq-limit-EVT}. Note that \(y \in S^{i,n}_\gamma\) implies that \(1 +\gamma z^i_n(y) > 0\). 

Finally, we prove the existence of the distribution \(G_\gamma^{i,n}\) for \(i \in C\) and \(n \in \mathbb{N}_{>0}\). Note that \(G_\gamma(x)\) is the CDF of \(G_\gamma\), which is a right-continuous monotone function with \(\lim_{x \to -\infty} G_\gamma(x) = 0\) and \(\lim_{x \to \infty} G_\gamma(x) = 1\). These conditions still hold after replacing \(x\) with \(z^i_n(y)\) in \(G_\gamma(x)\), since \(z^i_n(y)\) is a linear function of \(y\) with the positive coefficient \(1/\theta_1 a_n \sqrt{E_i}\). Thus, there exists a distribution \(G_\gamma^{i,n}\) whose CDF is \(G_\gamma^{i,n}(y)\).


\bibliographystyle{IEEEtran}
\bibliography{paper-2-VF-EVT.bib}{}

\end{document}